\documentclass[12pt]{article}

\usepackage{amsmath}
\usepackage{amssymb}
\usepackage{amsfonts}
\usepackage{latexsym}
\usepackage{color}

\catcode `\@=11 \@addtoreset{equation}{section}

\catcode `\@=12

\newcommand{\be}{\begin{equation}}
\newcommand{\en}{\end{equation}}
\newcommand{\bea}{\begin{eqnarray}}
\newcommand{\ena}{\end{eqnarray}}
\newcommand{\beano}{\begin{eqnarray*}}
\newcommand{\enano}{\end{eqnarray*}}
\newcommand{\bee}{\begin{enumerate}}
\newcommand{\ene}{\end{enumerate}}

\newcommand{\N}{\mathfrak N}

\newcommand{\mc}{\mathcal}

\newcommand{\D}{{\mc D}}

\newcommand{\E}{{\cal E}}
\newcommand{\F}{{\cal F}}

\newcommand{\Lc}{{\cal L}}

\newcommand{\1}{1 \!\! 1}

\newcommand{\Hil}{\mc H}

\newtheorem{thm}{Theorem}

\newtheorem{prop}[thm]{Proposition}

\newenvironment{proof}{\noindent {\bf Proof --}}{\hfill$\square$ \vspace{3mm}\endtrivlist}

\catcode `\@=11 \@addtoreset{equation}{section}
\catcode `\@=12

\begin{document}

\thispagestyle{empty}

\vspace*{2cm}

\begin{center}
{\Large \bf Construction of pseudo-bosons systems}   \vspace{2cm}\\

{\large F. Bagarello}\\
  Dipartimento di Metodi e Modelli Matematici,
Facolt\`a di Ingegneria,\\ Universit\`a di Palermo, I-90128  Palermo, Italy\\
e-mail: bagarell@unipa.it

\end{center}

\vspace*{2cm}

\begin{abstract}
\noindent In a recent paper we have considered an  explicit model of a PT-symmetric system based on a modification of the canonical commutation relation. We have introduced the so-called {\em pseudo-bosons}, and the role of Riesz bases in this context has been analyzed in detail. In this paper we consider a general construction of pseudo-bosons based on an explicit { coordinate-representation}, extending what is usually done in ordinary supersymmetric quantum mechanics. We also discuss an example arising from a linear modification of standard creation and annihilation operators, and we analyze its connection with coherent states.

\end{abstract}

\vspace{2cm}

%{\bf PACS Numbers}:  .......

\vfill

%\pagenumbering{roman}

\newpage

\section{Introduction}

In a recent paper, \cite{tri}, Trifonov suggested a possible explicit model of a PT-symmetric system based on a modification of the canonical commutation relation (CCR). The physical relevance of this model, and of the {\em particles} the model describes, the so-called {\em pseudo-bosons}, follows from the fact that it provides a nice example of what is called {\em pseudo-hermitian quantum mechanics} (PHQM) in the sense discussed in \cite{mosta,mosta2,bender} and in references therein. In PHQM self-adjoint hamiltonians are replaced by operators satisfying certain rules with respect to the parity and the time reversal operators and, as a consequence, possess eigenvalues which are real or which appear in conjugate pairs. However, \cite{tri} neglects many mathematical details of the model, making most of its results  purely formal. In \cite{bag1} we have considered  the same abstract model, but adopting a  mathematically rigorous point of view.  In particular, this  analysis has produced a somehow unexpected result, showing that Riesz bases, \cite{you,chri}, play a crucial role in this context. in this paper we continue our analysis and we construct other examples of pseudo-bosons working in an explicit coordinate representation and taking $\Lc^2(\Bbb{R})$ as our Hilbert space. We will see that, under special conditions, the procedure considered here collapses into that discussed for ordinary supersymmetric (Susy) quantum mechanics.

 The paper is organized as follows: in the next section we introduce the problem and, to keep the paper self-contained, we summarize our previous results.  In Section III we show how a new class of examples can be constructed. Section IV contains some consequences of our construction, while  we consider few concrete examples in Section V. Section VI, finally, contains a rather different example arising from a linear modification of the CCR, which is interesting for us since produces some results on coherent states.

\section{Description of the system}

Let $\Hil$ be a given Hilbert space with scalar product $\left<.,.\right>$ and related norm $\|.\|$. In \cite{tri,bag1} two operators $a$ and $b$ acting on $\Hil$ and satisfying the following commutation rule
\be
[a,b]=\1
\label{21}
\en
were introduced. Of course, this reduces to the CCR if $b=a^\dagger$. It is well known that $a$ and $b$ cannot be both bounded operators, so that they cannot be defined in all of $\Hil$. In the rest of the paper, given a certain operator $X$, we will call $D(X)$ its domain. In \cite{bag1} we have considered the following

\vspace{2mm}

{\bf Assumption 1.--} there exists a non-zero $\varphi_0\in\Hil$ such that $a\varphi_0=0$ and $\varphi_0\in D^\infty(b):=\cap_{k\geq0}D(b^k)$.

\vspace{2mm}

Under this assumption we can introduce the  vectors
\be
\varphi_n=\frac{1}{\sqrt{n!}}\,b^n\,\varphi_0, \quad n\geq 0, \quad \mbox{or}\quad \varphi_n=\frac{1}{\sqrt{n}}\,b\,\varphi_{n-1}, \quad n\geq 1,
\label{22}\en
which clearly belong to $\Hil$ for all $n\geq 0$. Let us now define the  (unbounded) operator $N:=ba$. Notice that $N\neq N^\dagger$. It is possible to check that $\varphi_n$ belongs to $D(N)$ for all $n\geq0$, and that
\be
N\varphi_n=n\varphi_n, \quad n\geq0.
\label{23}\en
Let us now put $\N:=N^\dagger=a^\dagger b^\dagger$. Because of (\ref{21})  we find $[N,b]=b$, $[N,a]=-a$, $[\N,a^\dagger]=a^\dagger$, $[\N,b^\dagger]=-b^\dagger$ and, moreover
\be
[b^\dagger,a^\dagger]=\1,
\label{24}\en
which again coincides with the CCR if $b^\dagger=a$. Let us now consider the following
\vspace{2mm}

{\bf Assumption 2.--} there exists a non-zero $\Psi_0\in\Hil$ such that $b^\dagger\Psi_0=0$ and $\Psi_0\in D^\infty(a^\dagger):=\cap_{k\geq0}D((a^\dagger)^k)$.

\vspace{2mm}\noindent

Hence  we can define
\be
\Psi_n=\frac{1}{\sqrt{n!}}(a^\dagger)^n\Psi_0, \quad n\geq 0, \quad \mbox{or}\quad \Psi_n=\frac{1}{\sqrt{n}}(a^\dagger)\Psi_{n-1}, \quad n\geq 1,
\label{25}\en
which  belong to $\Hil$ for all $n\geq 0$. They also belong to the domain of $\N$ and
\be
\N\Psi_n=n\Psi_n, \quad n\geq0.
\label{26}
\en

\vspace{2mm}

In the above assumptions we have $\left<\Psi_n,\varphi_m\right>=\delta_{n,m}\left<\Psi_0,\varphi_0\right>$ for all $n, m\geq0$, which, if  $\left<\Psi_0,\varphi_0\right>=1$, becomes
\be
\left<\Psi_n,\varphi_m\right>=\delta_{n,m}, \quad \forall n,m\geq0.
\label{27}\en
This means that the $\Psi_n$'s and the $\varphi_n$'s are biorthogonal. Moreover we have shown in \cite{bag1} that $\varphi_n\in D(a)$ and $\Psi_n\in D(b^\dagger)$ for all $n\geq0$, and that
$
a\varphi_n=\left\{
    \begin{array}{ll}
        0,\hspace{1.9cm}\mbox{ if } n=0,  \\
        \sqrt{n}\,\varphi_{n-1}, \hspace{.6cm} \mbox{ if } n>0,\\
       \end{array}
        \right.
        $
and
$
b^\dagger\Psi_n=\left\{
    \begin{array}{ll}
        0,\hspace{1.9cm}\mbox{ if } n=0,  \\
        \sqrt{n}\,\Psi_{n-1}, \hspace{.6cm} \mbox{ if } n>0.\\
       \end{array}
        \right.
$

Calling $\D_\varphi$ and $\D_\Psi$ respectively the linear span of  $\F_\varphi=\{\varphi_n,\,n\geq 0\}$ and $\F_\Psi=\{\Psi_n,\,n\geq 0\}$, and $\Hil_\varphi$ and $\Hil_\Psi$ their closures, we can also prove that
\be
f=\sum_{n=0}^\infty \left<\Psi_n,f\right>\,\varphi_n, \quad \forall f\in\Hil_\varphi,\qquad  h=\sum_{n=0}^\infty \left<\varphi_n,f\right>\,\Psi_n, \quad \forall h\in\Hil_\Psi.
\label{210}\en
What is not in general ensured is that  $\Hil_\varphi=\Hil_\Psi=\Hil$. With our assumptions we can only state that $\Hil_\varphi\subseteq\Hil$ and $\Hil_\Psi\subseteq\Hil$. However, in all the examples considered in \cite{bag1}, these three Hilbert spaces really coincide and for this reason it is natural to consider the following

\vspace{2mm}

{\bf Assumption 3.--} The above Hilbert spaces  coincide: $\Hil_\varphi=\Hil_\Psi=\Hil$.

\vspace{2mm}

From (\ref{210}) we deduce now that both $\F_\varphi$ and $\F_\Psi$ are bases in $\Hil$. The resolution of the identity looks now
\be
\sum_{n=0}^\infty
|\varphi_n><\Psi_n|=\sum_{n=0}^\infty
|\Psi_n><\varphi_n|=\1,
\label{211}\en
where $\1$ is the  identity of $\Hil$ and where the useful Dirac bra-ket notation has been adopted. In \cite{tri} the following operators were introduced
\be
\eta_\varphi=\sum_{n=0}^\infty
|\varphi_n><\varphi_n|,\qquad \eta_\Psi=\sum_{n=0}^\infty
|\Psi_n><\Psi_n|.
\label{212}\en
However, neither $\eta_\varphi$  nor $\eta_\Psi$ need to be well defined: for instance these series could  be not convergent, or even if they converge, they could converge to some unbounded operator, so we have to be careful about domains. This is, in fact, what we have done in \cite{bag1}: $\eta_\varphi$ acts on a vector $f$ in its domain $D(\eta_\varphi)$ as $\eta_\varphi f=\sum_{n=0}^\infty\left<\varphi_n,f\right>\varphi_n$ and $\eta_\Psi$  acts on a vector $h$ in its domain $D(\eta_\Psi)$ as $\eta_\Psi h=\sum_{n=0}^\infty\left<\Psi_n,h\right>\Psi_n$. Under Assumption 3, both these operators are densely defined in $\Hil$. In particular, we find that
\be
\eta_\varphi\Psi_n=\varphi_n,\qquad
\eta_\Psi\varphi_n=\Psi_n,
\label{213}\en
for all $n\geq0$, which also implies that $\Psi_n=(\eta_\Psi\eta_\varphi)\Psi_n$ and $\varphi_n=(\eta_\varphi\eta_\Psi)\varphi_n$, for all $n\geq0$. Hence
\be
\eta_\Psi\eta_\varphi=\eta_\varphi\eta_\Psi=\1 \quad \Rightarrow \quad \eta_\Psi=\eta_\varphi^{-1}.
\label{214}\en
In other words, both $\eta_\Psi$ and $\eta_\varphi$ are invertible and one is the inverse of the other. Furthermore, they are both positive defined and symmetric. However they are not in general bounded. Indeed we know,  \cite{you}, that two biorthogonal bases are related by a bounded operator, with bounded inverse, if and only if they are Riesz bases. This is why we consider

\vspace{2mm}

{\bf Assumption 4.--} $\F_\varphi$ and $\F_\Psi$ are Bessel sequences. In other words, there exist two positive constants $A_\varphi,A_\Psi>0$ such that, for all $f\in\Hil$,
\be
\sum_{n=0}^\infty\,|\left<\varphi_n,f\right>|^2\leq A_\varphi\,\|f\|^2,\qquad \sum_{n=0}^\infty\,|\left<\Psi_n,f\right>|^2\leq A_\Psi\,\|f\|^2.
\label{215}\en

\vspace{2mm}
\noindent
This assumption is equivalent to require that $\F_\varphi$ and $\F_\Psi$ are both Riesz bases, and implies  that $\eta_\varphi$ and $\eta_\Psi$ are bounded operators:
$\|\eta_\varphi\|\leq A_\varphi$, $\|\eta_\Psi\|\leq A_\Psi$. Moreover
$
\frac{1}{A_\Psi}\,\1\leq \eta_\varphi \leq A_\varphi\,\1,$ and  $\frac{1}{A_\varphi}\,\1\leq \eta_\Psi \leq A_\Psi\,\1.
$
Hence the domains of $\eta_\varphi$ and $\eta_\Psi$ can be taken to be all of $\Hil$.

\vspace{2mm}

In \cite{tri,bag1} several examples of operators $a$ and $b$ satisfying  (\ref{21}) have been considered. They all arise from the standard annihilation and creation operators
$c:=\frac{1}{\sqrt{2}}\left(\frac{d}{dx}+x\right)$ and $c^\dagger=\frac{1}{\sqrt{2}}\left(-\frac{d}{dx}+x\right)$ on $\Hil=\Lc^2(\Bbb{R},dx)$, $[c,c^\dagger]=\1$, as follows:
\begin{itemize}
\item choice 1, the trivial one: $a=c$ and $b=c^\dagger$.

\item choice 2, a one-parameter deformation: $a_s=c+sc^\dagger$ and $b_s=sc+(1+s^2)c^\dagger$ for all real $-1<s<1$.

\item choice 3, a two-parameters deformation: $a_{\alpha,\mu}:=\alpha c+\frac{\alpha}{\mu} c^\dagger$, $b_{\alpha,\mu}:=\mu\frac{\alpha^2-1}{\alpha} c+\alpha c^\dagger$ for $\alpha>1$ and $1<\mu<1+\frac{1}{\alpha^2-1}$.

\end{itemize}

With these choices the first three assumptions can be easily checked, while the fourth one is clear for the trivial choice but was not discussed for choices 2 and 3. Notice that all these choices are linear in both $\frac{d}{dx}$ and $x$. 

Quite interestingly,  any Riesz basis produces a pair of operators $a$ and $b$ satisfying $[a,b]=\1$ and all the above assumptions, so that more examples of pseudo-bosons could be constructed from {\bf any} Riesz basis, \cite{bag1,bagcal}.

\section{A new class of examples}

In this section we take  $\Hil=\Lc^2(\Bbb R)$ and we look for solutions of the commutation rule in (\ref{21}) of the following form:
\be
a=\frac{1}{\sqrt{2}}\left(\frac{d}{dx}+W_a(x)\right), \qquad b=\frac{1}{\sqrt{2}}\left(-\frac{d}{dx}+W_b(x)\right).
\label{31}\en
Here $W_j(x)$, $j=a,b$, are two functions such that $W_a(x)\neq \overline{W_b(x)}$. Hence $b^\dagger\neq a$. For future convenience we will assume that both $W_a(x)$ and $W_b(x)$ are sufficiently regular functions, for example that they are differentiable. We will show how to fix these functions in such a way Assumptions 1-4  are satisfied, while explicit choices of $W_a(x)$ and $W_b(x)$ will be considered in Section V. The starting point is that $a$ and $b$ are required to satisfy (\ref{21}): $[a,b]=\1$. A straightforward computation implies that $W_a(x)$ and $W_b(x)$ must satisfy the following simple equality:
\be
W_a(x)+W_b(x)=2x+\alpha,
\label{32}\en
where $\alpha$ is an arbitrary complex integration constant.  In particular, if we compute $N=ba$ and we use (\ref{32}) we get
\be
N=b\,a=\frac{1}{2}\left(-\frac{d^2}{dx^2}+U(x)\,\frac{d}{dx}+V(x)\right),
\label{33}\en
where $V(x):=W_a(x)(2x+\alpha-W_a(x))-W_a'(x)$ and $U(x):=2x+\alpha-2 W_a(x)$.

We observe that the approach we are adopting here is just an extension of the standard ideas of SUSY quantum mechanics, see \cite{CKS,jun} for a nice review. In Susy quantum mechanics the operator $N$ is just the hamiltonian of the system and $W(x)=W_a(x)=W_b(x)$ is the so-called super-potential, which is related to the (physical) potential via a Riccati equation. For this reason we still call both $W_a(x)$ and $W_b(x)$ superpotentials. In the first part of this section we will limit ourselves to real functions $W_a(x)$ and $W_b(x)$, extending our results to complex superpotentials in the second part. This will produce some interesting results, as we will see. Hence $\alpha$ in (\ref{32}) will be taken to be real, for the moment.

\vspace{2mm}

{\bf Remark:--} It may be interesting to observe that if $U(x)\equiv0$, then $N$ in (\ref{33}) looks like a one-dimensional hamiltonian (at least formally: we should  check for self-adjointness of the operator). This choice produces a  well known situation: $U(x)=0$ implies that $W_a(x)=x+\frac{\alpha}{2}$ and $V(x)=\left(x+\frac{\alpha}{2}\right)^2-1$ so that $N$ becomes, but for an unessential constant, the hamiltonian of a shifted harmonic oscillator,  $N=\frac{1}{2}\left(-\frac{d^2}{dx^2}+\left(x+\frac{\alpha}{2}\right)^2-1\right)$.This is in agreement with the fact that $W_b(x)=2x+\alpha-W_a(x)=W_a(x)$. Hence, if $\alpha$ is real, we deduce that $a^\dagger=b$ and we recover the ordinary CCR.

\vspace{2mm}

The next step consists in solving the two equations $a\varphi_0(x)=0$ and $b^\dagger\Phi_0(x)=0$, looking for solutions in $\Hil=\Lc^2(\Bbb R)$. These solutions are easily found: \be\varphi_0(x)=N_\varphi\exp\{-w_a(x)\},\qquad \Psi_0(x)=N_\Psi\exp\{-w_b(x)\},\label{34}\en where $N_\varphi$ and $N_\Psi$ are normalization constants. We have introduced here the following functions
\be
w_j(x)=\int W_j(x)\,dx,
\label{35}\en
$j=a,b$. The normalization constants can be written as $N_\varphi=\varphi_0(0)\,\exp\{w_a(0)\}$ and $N_\Psi=\Psi_0(0)\,\exp\{w_b(0)\}$. Of course since $\varphi_0(x)$ and $\Psi_0(x)$ must be square integrable, this imposes some constraints on the asymptotic behaviors of the $w_j(x)$'s and, as a consequence,  on the $W_j(x)$'s. We will consider this aspect in more details below, when  checking  that $\varphi_0(x)$ belongs to  $D^\infty(b)$, and that $\Psi_0(x)$ belongs to   $D^\infty(a^\dagger)$.

It is possible to prove that, independently of the analytic expressions of the $w_j(x)$'s, the following is true: $b^n\varphi_0(x)$ is proportional to a certain polynomial of degree $n$, $p_n(x)$, times  $\exp\{-w_a(x)\}$. In the same way we can also check that $(a^\dagger)^n\Psi_0(x)$ is proportional to a second polynomial of degree $n$, $q_n(x)$, times $\exp\{-w_b(x)\}$. Hence, if both $w_a(x)$ and $w_b(x)$ diverges to $+\infty$ when $|x|\rightarrow\infty$ at least as $|x|^\mu$ for some positive $\mu$, Assumptions 1 and 2 are satisfied.

More explicitly, if we define $\varphi_n(x)$ and $\Psi_n(x)$ as in (\ref{22}) and (\ref{25}), we can prove that
\be
\varphi_n(x)=N_n^\varphi\,p_n(x)\,\exp\{-w_a(x)\},\qquad N_n^\varphi=\frac{\varphi_0(0)\,\exp\{w_a(0)\}}{\sqrt{n!\,2^n}},
\label{36}\en
and
\be
\Psi_n(x)=N_n^\Psi\,p_n(x)\,\exp\{-w_b(x)\},\qquad N_n^\Psi=\frac{\Psi_0(0)\,\exp\{w_b(0)\}}{\sqrt{n!\,2^n}},
\label{37}\en
where an {\bf unique} polynomial $p_n(x)$ appears both in $\varphi_n(x)$ and in $\Psi_n(x)$. This is defined recursively as follows: $p_0(x)=1$ and $p_{n+1}(x)=(2x+\alpha)p_n(x)-p_n'(x)$, $n\geq 0$. Therefore $p_1(x)=2x+\alpha$, $p_2(x)=(2x+\alpha)^2-2$, $p_3(x)=(2x+\alpha)\left((2x+\alpha)^2-6\right)$ and so on. The proof of this claim is based on induction. Indeed, but for unessential constants, we have:
$$
b^{n+1}\varphi_0(x)\simeq b\left(p_n(x)\varphi_0(x)\right) \simeq -\frac{d}{dx}\left(p_n(x)\varphi_0(x)\right)+W_b(x)\left(p_n(x)\varphi_0(x)\right)\simeq
$$
$$
\simeq-p_n'(x)\,e^{-w_a(x)}-p_n(x)\frac{d}{dx}\,e^{-w_a(x)}+W_b(x)p_n(x)\,e^{-w_a(x)}=
\left(-p_n'(x)+(2x+\alpha)p_n(x)\right)\,e^{-w_a(x)}.
$$
To the same  conclusion we  arrive computing $(a^\dagger)^n\Psi_0(x)$.

Notice that, because of (\ref{32}), we also have that
\be
w_a(x)+w_b(x)=x^2+\alpha x+\beta,
\label{38}\en
where $\beta$ is a second integration constant which again we take real for the moment. Therefore, since each one of the functions $w_j(x)$ should diverge to $+\infty$ for large $|x|$ as $|x|^{\mu_j}$ for some $\mu_j>0$, $j=a,b$, this equality also fixes an upper bound for the $\mu_j$'s: we must have $0<\mu_j\leq2$, $j=a,b$.

Using (\ref{23}) and (\ref{26}) we have
\be
N\varphi_n(x)=n\,\varphi_n(x), \qquad N^\dagger \Psi_n(x)=a^\dagger\,b^\dagger \Psi_n(x)=n\Psi_n(x),
\label{39}\en
for all $n\geq0$. Moreover these functions are biorthogonal:
\be
\left<\varphi_n,\Psi_m\right>=\delta_{n,m}\left<\varphi_0,\Psi_0\right>,
\label{310}\en
i.e.,
\be
\sqrt{\frac{1}{n!m!2^{n+m}}}\int_{\Bbb{R}}p_n(x)p_m(x)e^{-(x^2+\alpha x+\beta)}\,dx=\delta_{n,m}\int_{\Bbb{R}}e^{-(x^2+\alpha x+\beta)}\,dx=
\delta_{n,m}\,\sqrt{\pi}\,e^{\alpha^2/4-\beta}
\label{311}\en

\vspace{3mm}

{\bf Remark:--} Our $p_n(x)$ are  related to Hermite polynomials since we can check that $p_n(x)=(-1)^ne^{x^2+\alpha x}\,\frac{d^n}{dx^n}\,e^{-(x^2+\alpha x)}$, for all $n\geq 0$. Furthermore, using this formula, is a standard computation to check  (\ref{311}) directly. In particular, it is simple to check that $\left<\varphi_n,\Psi_m\right>=0$ if $n\neq m$.

\vspace{2mm}

We are now ready to check if or when Assumption 3 is verified. For that we introduce as in Section II the sets $\F_\varphi=\left\{\varphi_n(x),\,n\geq0\right\}$ and $\F_\Psi=\left\{\Psi_n(x),\,n\geq0\right\}$, and we construct $\D_\varphi$ and $\D_\Psi$, which are respectively the linear span of  $\F_\varphi$ and $\F_\Psi$, and their closures $\Hil_\varphi$ and $\Hil_\Psi$. Hence, by construction, $\F_\varphi$ is complete in $\Hil_\varphi$ and $\F_\Psi$ is complete in $\Hil_\Psi$. We need to check whether   $\Hil_\varphi=\Hil_\Psi=\Hil$.

To check this we first observe that $\F_\varphi$ is complete in $\Hil$ if and only if the set $\F_\pi^{(a)}=\left\{\pi_n^{(a)}(x):=x^n\,e^{-w_a(x)},\,n\geq0\right\}$ is complete in $\Hil$. Analogously, $\F_\Psi$ is complete in $\Hil$ if and only if the set $\F_\pi^{(b)}=\left\{\pi_n^{(b)}(x):=x^n\,e^{-w_b(x)},\,n\geq0\right\}$ is complete in $\Hil$. But, \cite{kolfom}, if $\rho(x)$ is a Lebesgue-measurable function which is different from zero almost everywhere (a.e.) in $\Bbb R$ and if there exist two positive constants $\delta, C$ such that $|\rho(x)|\leq C\,e^{-\delta|x|}$ a.e. in $\Bbb R$, then the set $\left\{x^n\,\rho(x)\right\}$ is complete in $\Lc^2(\Bbb{R})$.

This suggests to consider the following constraint on the asymptotic behavior of the $w_j(x)$'s: for Assumption 3 to be satisfied it is sufficient that four positive constants $C_j,\,\delta_j$, $j=a,b$ exist such that
\be
\left|e^{-w_j(x)}\right|\leq C_j\,e^{-\delta_j|x|},
\label{311b}
 \en
 $j=a,b$, holds a.e. in $\Bbb R$. It should be noticed that this condition is stronger than the one required for Assumptions 1 and 2 to hold, since for instance is not satisfied if $w_a(x)\simeq|x|^{1/2}$ for large $|x|$.

Using now the biorthogonality of the sets $\F_\varphi$ and $\F_\Psi$, and their completeness in $\Lc^2(\Bbb R)$, it is now clear that, given any function $f(x)\in \Lc^2(\Bbb R)$,
\be
f(x)=\frac{1}{\left<\Psi_0,\varphi_0\right>}\,\sum_{k=0}^\infty\,\left<\Psi_k,f\right>\varphi_k(x)=
\frac{1}{\left<\varphi_0,\Psi_0\right>}\,\sum_{k=0}^\infty\,\left<\varphi_k,f\right>\Psi_k(x).
\label{312}\en
This can also be written in the usual bra-ket notation as in (\ref{211}):
\be
\frac{1}{\left<\Psi_0,\varphi_0\right>}\,\sum_{k=0}^\infty
|\varphi_k><\Psi_k|=\frac{1}{\left<\varphi_0,\Psi_0\right>}\,\sum_{k=0}^\infty
|\Psi_k><\varphi_k|=\1,
\label{313}\en
where $\1$ is the identity operator on $\Lc^2(\Bbb R)$. The overall constants $\left<\Psi_0,\varphi_0\right>^{-1}$ and $\left<\varphi_0,\Psi_0\right>^{-1}$ appear because of (\ref{310}).

Suppose now that we are interested in going from $\F_\varphi$ to $\F_\Psi$ and viceversa. In other words we are now interested to introduce an invertible operator $S$ mapping each $\varphi_n$ into $\Psi_n$, $S\varphi_n=\Psi_n$, whose inverse of course satisfies $S^{-1}\Psi_n=\varphi_n$, for all $n\geq 0$. As we have already discussed, both $S$ and $S^{-1}$ may be unbounded, so a special care is required. A formal expansion of these operators is
\be
S=\frac{1}{\left<\Psi_0,\varphi_0\right>}\,\sum_{k=0}^\infty
|\Psi_k><\Psi_k|,\qquad S^{-1}=\frac{1}{\left<\varphi_0,\Psi_0\right>}\,\sum_{k=0}^\infty
|\varphi_k><\varphi_k|.
 \label{314}\en
It is quite easy to check that, again at least formally, $SS^{-1}=S^{-1}S=\1$. Due to the analytic expressions (\ref{36}) and (\ref{37}) of our wave-functions $\varphi_n(x)$ and $\Psi_n(x)$, we  deduce that
\be
S=\frac{\Psi_0(0)}{\varphi_0(0)}\,\frac{e^{\delta w_a(x)}}{e^{\delta w_b(x)}},\qquad
S^{-1}=\frac{\varphi_0(0)}{\Psi_0(0)}\,\frac{e^{\delta w_b(x)}}{e^{\delta w_a(x)}},
\label{315}\en
where we have introduced $\delta w_j(x):=w_j(x)-w_j(0)$, $j=a,b$. A sufficient condition for both $S$ and $S^{-1}$ to be bounded operators from $\Lc^2(\Bbb R)$ into itself is now easily deduced using  equation (\ref{38}), which implies that $\frac{e^{\delta w_a(x)}}{e^{\delta w_b(x)}}=\frac{e^{2\delta w_a(x)}}{e^{x^2+\alpha x}}$ and $\frac{e^{\delta w_b(x)}}{e^{\delta w_a(x)}}=\frac{e^{x^2+\alpha x}}{e^{2\delta w_a(x)}}$:

{\em if $\frac{e^{2\delta w_a(x)}}{e^{x^2+\alpha x}}\in \Lc^\infty(\Bbb R)$, then $S\in B(\Lc^2(\Bbb R))$. Moreover, if $\frac{e^{x^2+\alpha x}}{e^{2\delta w_a(x)}}\in \Lc^\infty(\Bbb R)$,  also $S^{-1}\in B(\Lc^2(\Bbb R))$.}

We recall that the existence of such an operator is equivalent to the fact that both $\F_\varphi$ and $\F_\Psi$ are Riesz bases, see \cite{you,bag1}. It is clear, however, that the above boundedness assumption imposes further limitations on the functions $w_j(x)$'s and, as a consequence, on the $W_j(x)$'s  defining  $a$ and $b$. For this reason in  Section V we will consider examples in which this last requirement is not satisfied, so that $\F_\varphi$ and $\F_\Psi$ are biorthogonal (but not necessarily Riesz) bases of $\Lc^2(\Bbb R)$, and other examples in which they do are Riesz bases since they are related by a bounded operator with bounded inverse. A similar situation will be discussed in Section VI in a slightly different context: we will deduce a sufficient condition for $\F_\varphi$ and $\F_\Psi$ to be Riesz bases, condition which is related to two families of related coherent states.

\subsection{What if the superpotentials are complex?}

The above result on the boundedness of $S$ and $S^{-1}$ displays the relevance of $\alpha$: suppose $\alpha\neq0$. If $\delta w_a(x)$ behaves as $x^2$ for large $|x|$ then $\frac{e^{2\delta w_a(x)}}{e^{x^2+\alpha x}}$ and $\frac{e^{x^2+\alpha x}}{e^{2\delta w_a(x)}}$ cannot be bounded for both positive and negative $x$. This is not true if $\alpha$ is purely imaginary, of course: both these fractions are bounded functions so that $S$ and $S^{-1}$ are  bounded operators. That's why this choice is so interesting for us. In this case formulas (\ref{36}) and (\ref{37}) look like
\be
\varphi_n(x)=N_n^\varphi\,p_n(x)\,\exp\{-w_a(x)\},\qquad N_n^\varphi=\frac{\varphi_0(0)\,\exp\{w_a(0)\}}{\sqrt{n!\,2^n}},
\label{316}\en
and
\be
\Psi_n(x)=N_n^\Psi\,\overline{p_n(x)}\,\exp\{-\overline{w_b(x)}\},\qquad N_n^\Psi=\frac{\Psi_0(0)\,\exp\{\overline{w_b(0)}\}}{\sqrt{n!\,2^n}},
\label{317}\en
where $p_n(x)$ is defined as before.   Next we find that
\be
\left<\varphi_n,\Psi_m\right>=\delta_{n,m}\left<\varphi_0,\Psi_0\right>=\delta_{n,m}\,\sqrt{\pi}\,\Psi_0(0)\,\overline{\varphi_0(0)}\,e^{\overline{\alpha}^2/4}
\label{318}\en
The main difference arises in the analytic expression of $S$ and of $S^{-1}$. For that it is necessary to introduce the operator of complex conjugation $C$ which acts on a generic function $f(x)\in\Lc^2(\Bbb R)$ as follows: $Cf(x)=\overline{f(x)}$. $C$ is antilinear and idempotent: $C^2=\1$. Hence $C=C^{-1}$. While formulas  (\ref{314}) are still true,  (\ref{315}) must be replaced by
\be
S=C\,\frac{\overline{\Psi_0(0)}}{\varphi_0(0)}\,\frac{e^{\delta w_a(x)}}{e^{\delta w_b(x)}},\qquad
S^{-1}=\frac{\varphi_0(0)}{\overline{\Psi_0(0)}}\,\frac{e^{\delta w_b(x)}}{e^{\delta w_a(x)}}\,C,
\label{319}\en
It is a straightforward computation to check that they are indeed the inverse of one another and that $S\varphi_n(x)=\Psi_n(x)$, $S^{-1}\Psi_n(x)=\varphi_n(x)$ for all $n\geq0$. As for the norms of $S$ and $S^{-1}$, they are not affected by the presence of $C$ and of the complex conjugation in $\Psi_0(0)$. For this reason the same conditions as above are recovered: $S$ and $S^{-1}$ are bounded if both $\frac{e^{2\delta w_a(x)}}{e^{x^2+\alpha x}}$ and  $\frac{e^{x^2+\alpha x}}{e^{2\delta w_a(x)}}$ belong to $\Lc^\infty(\Bbb R)$. This means that, if $\alpha$ is purely imaginary, then both $S$ and $S^{-1}$ can be bounded and, as a consequence, $\F_\varphi$ and $\F_\Psi$ are Riesz bases. Once again we stress that, if $\alpha$ is real, this is never possible.

\vspace{4mm}

Due to the explicit form of the operator $S$ we have deduced before it is now interesting to consider the orthonormal basis arising, for instance, from the action of $S^{-1/2}$ onto the Riesz basis $\F_\varphi$: under the Assumptions 1-4 of Section 2, the functions $\hat\varphi_n(x):=S^{-1/2}\varphi_n(x)$ give indeed an orthonormal basis of $\Lc^2(\Bbb R)$. The analytic expression of these vectors is easily found at least for real superpotentials, while it is  less evident when $W_a(x)$ and $W_b(x)$ are complex. In this first case, using (\ref{315}), we find that
\be
\hat\varphi_n(x)=\sqrt{\frac{\Psi_0(0)\varphi_0(0)}{2^n\,n!}}\,p_n(x)\,e^{-\frac{1}{2}(x^2+\alpha x)}.
\label{48}
\en
However, if $W_a(x)$ and $W_b(x)$ are real, we have already seen that $S$ and/or $S^{-1}$ are unbounded so that a certain care in the definition of, say, $S^{-1/2}$ is required. However, equation (\ref{48}) holds true since $\varphi_n(x)$ belongs to the domain of $S^{-1/2}$, which turns out to be well defined.
It is not hard to see that these functions reduce to the standard Hermite functions under suitable conditions, see also Example 1 of  Section V.

\section{Consequences of our construction}

In \cite{bag1} we have seen how pseudo-bosons are related to coherent states, intertwining operators and PHQM. In this section we will see how  these relations look like in this present settings.

First of all it is possible to check that, if $S$ and $S^{-1}$ are both bounded and self-adjoint,
\be
b=S^{-1}a^\dagger S,\qquad b^\dagger=S\,a\,S^{-1}.
\label{41}
\en
Of course from (\ref{41}) we also deduce that $a=S^{-1}b^\dagger S$ and  $a^\dagger=S\,b\,S^{-1}$. These equalities imply the following intertwining equations:
\be
S\,N=\N\,S,\qquad N\,S^{-1}=S^{-1}\N
\label{42}\en
 which, of course, are in agreement with the fact that $N$ and $\N$ are isospectrals and that their eigenstates are related by $S$ via the equation $S\varphi_n(x)=\Psi_n(x)$.  We refer to \cite{intop,bagintop} for more results on intertwining operators.
 As  noticed in \cite{bag1},  condition (\ref{42})  states that $N$ and $\N$ are pseudo-hermitian conjugate via $S$, \cite{mosta}. We recall that this was just the main motivation in \cite{tri} for considering the  commutation rules in (\ref{21}).

\vspace{3mm}

Under Assumptions 1-4, some kind of {\em bi-coherent states} can be introduced, \cite{bag1}. Let us define the $z$-dependent operators
\be
U(z)=\exp\{z\,b-\overline{z}\,a\}, \qquad V(z)=\exp\{z\,a^\dagger-\overline{z}\,b^\dagger\},
\label{43}\en
$z\in\Bbb{C}$, and the following vectors:
\be
\varphi(z)=U(z)\varphi_0=e^{-|z|^2/2}\,\sum_{n=0}^\infty\,\frac{z^n}{\sqrt{n!}}\,\varphi_n,\qquad \Psi(z)=V(z)\,\Psi_0=e^{-|z|^2/2}\,\sum_{n=0}^\infty\,\frac{z^n}{\sqrt{n!}}\,\Psi_n.
\label{44}\en
Both these series are convergent for all possible $z\in\Bbb{C}$ due to the fact that $S$ and $S^{-1}$ are bounded, \cite{bag1}. These vectors are called {\em coherent} since they are eigenstates of some lowering operators. Indeed we can check that
\be
a\varphi(z)=z\varphi(z), \qquad b^\dagger\Psi(z)=z\Psi(z),
\label{45}\en
for all $z\in\Bbb{C}$. Moreover we have
\be
\frac{1}{\pi}\int_{\Bbb{C}}\,dz |\varphi(z)><\varphi(z)|=S^{-1}, \qquad
\frac{1}{\pi}\int_{\Bbb{C}}\,dz |\Psi(z)><\Psi(z)|=S,
\label{46}\en
and
\be
\frac{1}{\pi}\int_{\Bbb{C}}\,dz |\varphi(z)><\Psi(z)|=
\frac{1}{\pi}\int_{\Bbb{C}}\,dz |\Psi(z)><\varphi(z)|=\1,
\label{47}\en
and this is why we call them bi-coherent. They can be associated to {\em standard} coherent states (i.e. coherent states built out of an orthonormal basis) if $S$ and $S^{-1}$ are bounded, because of the properties of Riesz bases. We don't give the details of this construction here since they are discussed in \cite{bag1}. In Section VI we will show that (\ref{47}) can be used to check whether $\F_\varphi$ and $\F_\Psi$ are Riesz bases or not, regardless of any information on $S$ and $S^{-1}$.

\section{Explicit examples}

We will now discuss three examples of our construction showing how easily Riesz bases can be constructed using a sort of perturbation technique applied to the harmonic oscillator. We will also consider an example which at a first sight seems to work but, because of a mathematical detail which should be properly considered, doesn't work at all.

\vspace{2mm}

{\bf Example 1:} we fix here $W_a(x)=x$. Hence $W_b(x)$ is fixed as in (\ref{32}) just requiring that the related operators $a$ and $b$, see (\ref{31}), satisfy $[a,b]=\1$. Hence $W_b(x)=x+\alpha$ where, for the moment, we don't make any assumption on $\alpha$. Then we get $w_a(x)=\frac{x^2}{2}+k_a$ and $w_b(x)=\frac{x^2}{2}+\alpha x+k_b$. Here $k_a$ and $k_b$ are two integration constants which are, in general, complex. Their sum gives back $\beta$, see (\ref{38}).

Using the inequality $e^{-x^2/2}\leq 2 e^{-|x|}$ it is  clear that $\left|e^{-w_a(x)}\right|\leq 2 \left|e^{-k_a}\right|\,e^{-|x|}$. Hence the set $\F_\varphi$ is a basis of $\Lc^2(\Bbb R)$. The same estimate, with $k_a$ replaced by $k_b$, can be repeated for $\left|e^{-w_b(x)}\right|$ if $\alpha$ is purely imaginary. If $\alpha$ is real this estimate does not work. However we get that $\left|e^{-w_b(x)}\right|\leq 2 \left|e^{-k_b}\right|\,e^{\alpha^2/2}\,e^{|\alpha|}\,e^{-|x|}$, which again implies that $\F_\Psi$ is a basis of $\Lc^2(\Bbb R)$.

A major difference arises if we require to these sets to be Riesz bases. Indeed, if $\alpha$ is purely imaginary,  $\left|\frac{e^{2\delta w_a(x)}}{e^{x^2+\alpha x}}\right|=\left|\frac{e^{x^2+\alpha x}}{e^{2\delta w_a(x)}}\right|=1$, so that both $S$ and $S^{-1}$ are bounded operators and $\F_\varphi$ and $\F_\Psi$ are automatically Riesz bases. If we rather look for real $\alpha$ such that the above fractions are both bounded functions, then the only possible choice is $\alpha=0$. Under this constraint the set of vectors in (\ref{48}) is nothing but the standard Hermite functions (at most but for an unessential overall phase). This is not surprising since, if $\alpha=0$, then $W_a(x)=W_b(x)$ and $a=b^\dagger$: we go back to the standard canonical commutation relation.

\vspace{2mm}

{\bf Example 2:} our above mentioned perturbation technique consists in adding a suitable perturbation to a {\em zero order} superpotential  $W_a^o(x)=x$. In particular we take here $W_a(x)=x+\cos(x)$. Hence, by (\ref{32}), $W_b(x)=x-\cos(x)+\alpha$. Consequently we have $w_a(x)=\frac{x^2}{2}+\sin(x)+k_a$ and $w_b(x)=\frac{x^2}{2}-\sin(x)+\alpha x+k_b$.

With the same considerations as above we can prove that, for all $x\in\Bbb R$,  $\left|e^{-w_a(x)}\right|\leq 2 \left|e^{1-k_a}\right|\,e^{-|x|}$ and  $\left|e^{-w_b(x)}\right|\leq 2 \left|e^{1-k_b}\,e^{\alpha^2/2}\right|\,e^{|\alpha|}\,e^{-|x|}$. Hence both $\F_\varphi$ and $\F_\Psi$ are bases for $\Lc^2(\Bbb R)$, independently of the nature of $\alpha$. However, if we want these to be Riesz bases, again a sufficient condition is that $\alpha$ is purely imaginary. Indeed with this choice both $\left|\frac{e^{2\delta w_a(x)}}{e^{x^2+\alpha x}}\right|$ and $\left|\frac{e^{x^2+\alpha x}}{e^{2\delta w_a(x)}}\right|$ are bounded functions, as desired. The operators $a$ and $b$ in (\ref{31}) are
$$
a=\frac{1}{\sqrt{2}}\left(\frac{d}{dx}+x+\cos(x)\right), \qquad b=\frac{1}{\sqrt{2}}\left(-\frac{d}{dx}+x-\cos(x)+i\alpha_r\right),
$$
for any fixed real $\alpha_r$.

\vspace{2mm}

{\bf Example 3:} Example 2 is a particular case of a rather more general situation which can be constructed by considering a function $\Phi(x)$ which is differentiable and bounded in $\Bbb R$: $-\infty<\Phi_m\leq \Phi(x)\leq \Phi_M<\infty$, $\forall x\in \Bbb R$. Now we define $W_a(x)=x+\Phi'(x)$. Hence, by (\ref{32}), $W_b(x)=x-\Phi'(x)+\alpha$. Consequently we have $w_a(x)=\frac{x^2}{2}+\Phi(x)+k_a$ and $w_b(x)=\frac{x^2}{2}-\Phi(x)+\alpha x+k_b$. The following inequalities hold:  $\left|e^{-w_a(x)}\right|\leq 2 \left|e^{-k_a}\right|\,e^{-\Phi_m}\,e^{-|x|}$ and  $\left|e^{-w_b(x)}\right|\leq 2 e^{\Phi_M} \left|e^{-k_b}\,e^{\alpha^2/2}\right|\,e^{|\alpha|}\,e^{-|x|}$. Therefore both $\F_\varphi$ and $\F_\Psi$ are bases for $\Lc^2(\Bbb R)$, independently of the nature of $\alpha$. As before, however, if $\alpha$ is purely imaginary then these are also Riesz bases, for the usual reason:  both $\left|\frac{e^{2\delta w_a(x)}}{e^{x^2+\alpha x}}\right|$ and $\left|\frac{e^{x^2+\alpha x}}{e^{2\delta w_a(x)}}\right|$ are bounded functions, as desired. The operators $a$ and $b$ in (\ref{31}) are
$$
a=\frac{1}{\sqrt{2}}\left(\frac{d}{dx}+x+\Phi'(x)\right), \qquad b=\frac{1}{\sqrt{2}}\left(-\frac{d}{dx}+x-\Phi'(x)+i\alpha_r\right),
$$
where $\alpha_r$ is an arbitrary but fixed real quantity.

A choice of $\Phi(x)$ which is not bounded but still {\em under control} is $\Phi(x)=\frac{\alpha x}{2}$. This produces $W_a(x)=W_b(x)=x+\frac{\alpha}{2}$, which is nothing but the shifted harmonic oscillator.

\vspace{2mm}

{\bf Example 4:} It may seem reasonable and interesting to replace the Hilbert space considered so far, $\Lc^2(\Bbb R)$, with another Hilbert space of functions defined on a bounded domain: $\Hil=\Lc^2(X)$, where $X=[l,L]$ and $L-l<\infty$. Suppose now that $a$ and $b$ are defined as in (\ref{31}) and that $D(a)=\{f(x)\in\Lc^2(X):\,f'(x)+W_a(x)f(x)\in\Lc^2(X)\}$ and $D(b)=\{f(x)\in\Lc^2(X):\,-f'(x)+W_b(x)f(x)\in\Lc^2(X)\}$. The adjoint of $a$ and $b$ can be computed with standard techniques and it turns out, in particular, that $b^\dagger=\frac{1}{\sqrt{2}}\left(\frac{d}{dx}+\overline{W_b(x)}\right)$ with $D(b^\dagger)=\{f(x)\in\Lc^2(X):\,f'(x)+\overline{W_b(x)}f(x)\in\Lc^2(X), \mbox{ and } f(l)=f(L)=0\}$. Notice that in the first three examples of this section, where we have essentially $l=-\infty$ and $L=\infty$, $f(l)=f(L)=0$ was automatically satisfied because of the asymptotic behavior  of any differentiable functions of $\Lc^2(\Bbb R)$. Now, in order to verify Assumption 2, we should find a non-zero function $\Psi_0(x)$ in the domain of $b^\dagger$ which is annihilated by $b^\dagger$ itself. But this is impossible since the only function which satisfies $b^\dagger\Psi_0(x)=0$ is $\Psi_0(x)=N_0^\Psi\,\exp\{-w_b(x)\}$, which cannot be zero in $l$ and $L$ except if it is identically zero.  So Assumption 2 is violated here, while Assumption 1 still holds true.

\section{A different example}

The example which we consider here is motivated by the paper \cite{sun}, where the author consider a  simple modification of the CCR  in connection with non-hermitian quantum systems.  The starting point is a lowering operator $a$ acting on an Hilbert space $\Hil$ which, together with its adjoint $a^\dagger$, satisfies the CCR $[a,a^\dagger]=\1$. Let us now consider the following simple deformation of $a$ and $a^\dagger$:
$$
A_\alpha=a-\alpha\,\1,\qquad B_\beta=a^\dagger-\beta\,\1.
$$
It is clear that $[A_\alpha,B_\beta]=\1$ and that, if $\alpha\neq\overline\beta$, $A_\alpha\neq B_\beta^\dagger$. To check Assumption 1 first of all we have to find a vector $\varphi_0(\alpha)$ such that $A_\alpha\varphi_0(\alpha)=0$. Such a vector clearly exists since $A_\alpha\varphi_0(\alpha)=0$ can be written as $a\varphi_0(\alpha)=\alpha\varphi_0(\alpha)$. Hence it is enough to take $\varphi_0(\alpha)$ as the following coherent state:
$$
\varphi_0(\alpha)=U(\alpha)\varphi_0,
$$
where $U(\alpha)=e^{\alpha a^\dagger-\overline\alpha a}=e^{-|\alpha|^2/2}e^{\alpha a^\dagger} e^{\overline{\alpha} a}$ and $\varphi_0$ is the vacuum of $a$: $a\varphi_0=0$. Incidentally we recall that the set $\E=\{\varphi_n=\frac{(a^\dagger)^n}{\sqrt{n!}}\varphi_0,\,n\geq 0\}$ is an orthonormal basis of $\Hil$. The fact that $\varphi_0(\alpha)$ belongs to $D^\infty(B_\beta)$ follows from the following estimate:  $$\|B_\beta^l\varphi_0(\alpha)\|\leq l! e^{|\overline\alpha-\beta|},$$ which holds for all $l\geq0$.

Let us now define a second vector $\Psi_0(\beta):=U(\overline\beta)\varphi_0$. This is a second coherent state, labeled by $\beta$, which satisfies Assumption 2: $B_\beta^\dagger\Psi_0(\beta)=0$ and $\Psi_0(\beta)\in D^\infty(A_\alpha^\dagger)$: as before we get $\|(A_\alpha^\dagger)^l\Psi_0(\beta)\|\leq l! e^{|\overline\alpha-\beta|}$, for all $l\geq0$.

Now we introduce, following (\ref{22}) and (\ref{25}), the vectors
\be
\varphi_n(\alpha,\beta):=\frac{B_\beta^n}{\sqrt{n!}}\,\varphi_0(\alpha), \qquad \Psi_n(\alpha,\beta):=\frac{(A_\alpha^\dagger)^n}{\sqrt{n!}}\,\Psi_0(\beta).
\label{61}\en
It is possible to rewrite $\varphi_n(\alpha,\beta)$ and $\Psi_n(\alpha,\beta)$ in many different equivalent forms. For instance we have
\be
\varphi_n(\alpha,\beta)=V_\varphi(\alpha,\beta)\varphi_n, \qquad V_\varphi(\alpha,\beta)=e^{-|\alpha|^2/2}e^{\alpha a^\dagger}e^{-\beta a}=e^{\alpha(\beta-\overline\alpha)/2}e^{\alpha a^\dagger-\beta a}
\label{62}\en
and
\be
\Psi_n(\alpha,\beta)=V_\Psi(\alpha,\beta)\varphi_n, \qquad V_\Psi(\alpha,\beta)=e^{-|\beta|^2/2}e^{\overline\beta a^\dagger}e^{-\overline\alpha a}=e^{\overline\beta (\overline\alpha-\beta)/2}e^{\overline\beta  a^\dagger-\overline\alpha  a},
\label{63}\en
for all $n\geq0$. The operators $V_\varphi$ and $V_\Psi$, which are in general unbounded (see below), are related by
\be
V_\Psi^\dagger(\alpha,\beta)=V_\varphi^{-1}(\alpha,\beta)\,
\exp\left\{\alpha\beta-\frac{1}{2}(|\alpha|^2+|\beta|^2)\right\}.
\label{64}\en
Notice that they are densely defined in $\Hil$ since each $\varphi_n$ belongs to $D(V_\varphi)$ and $D(V_\Psi)$.

\vspace{2mm}

{\bf Remark:--} It is interesting to notice that, if $\beta=\overline\alpha$, then everything collapses: $B_\beta^\dagger=A_\alpha$, $\varphi_0(\alpha)=\Psi_0(\beta)$, $\varphi_n(\alpha,\beta)=\Psi_n(\alpha,\beta)$ and, finally, $V_\varphi$ and $V_\Psi$ are unitary operators.\vspace{2mm}

Defining as usual $\F_\varphi^{(\alpha,\beta)}=\{\varphi_n(\alpha,\beta), n\geq0\}$ and $\F_\Psi^{(\alpha,\beta)}=\{\Psi_n(\alpha,\beta), n\geq0\}$, it is possible to check that both these sets are complete in $\Hil$. This is a subtle point: indeed it is quite easy to prove for instance that, if $f\in D(V_\varphi)$ is orthogonal to all the $\varphi_n(\alpha,\beta)$, $n\geq0$, then $f=0$. However, this does not necessarily  implies that taken $h\in\Hil$, $h\notin D(V_\varphi)$, such that $\left<h,\varphi_n(\alpha,\beta)\right>=0$ for all $n\geq0$,  then $h=0$, even if $D(V_\varphi)$ is dense in $\Hil$. Therefore, to prove the completeness of $\F_\varphi$, it is convenient to rewrite $\varphi_n(\alpha,\beta)$, in the following equivalent way:
$$
\varphi_n(\alpha,\beta)=\frac{1}{\sqrt{n!}}\,e^{(\alpha\,\beta-\overline{\alpha}\,\overline{\beta})/2}\,
U(\overline{\beta}) (a^\dagger)^n U(\alpha-\overline{\beta})\varphi_0,
$$
and to use induction on $n$ and the properties of the unitary operators $U(\overline{\beta})$ and $U(\alpha-\overline{\beta})$. With the same techniques we can check that $\F_\Psi$ is complete in $\Hil$.

The vectors in $\F_\varphi$ and $\F_\Psi$ are also biorthogonal: using (\ref{62}), (\ref{63}) and (\ref{64}) we find
$$
\left<\varphi_n(\alpha,\beta),\Psi_m(\alpha,\beta)\right>=\left<V_\varphi(\alpha,\beta)\varphi_n,
V_\Psi(\alpha,\beta)\varphi_m\right>=$$
$$=\left<V_\Psi^\dagger(\alpha,\beta) V_\varphi(\alpha,\beta)\varphi_n,\varphi_m\right>=\delta_{n,m}\exp\{\overline\alpha\overline\beta-\frac{1}{2}(|\alpha|^2+|\beta|^2)\}.
$$
Of course biorthonormality could be recovered changing the normalization of $\varphi_0(\alpha,\beta)$ and $\Psi_0(\alpha,\beta)$.

As for Assumption 4, the situation is a bit more difficult:  if $\beta=\overline\alpha$, then both $\F_\varphi$ and $\F_\Psi$ are the same orthonormal basis. However, whenever $\beta\neq\overline\alpha$, it is possible to prove that neither $\F_\varphi$ nor $\F_\Psi$ can be Riesz bases (or, equivalently, Bessel sequences). Indeed, let us suppose, e.g., that $\F_\varphi$ is a Riesz basis. Then $\|\varphi_n(\alpha,\beta)\|$  must be uniformly bounded in $n$ by a constant related to the norm of the frame operator of $\F_\varphi$, \cite{bag1,bagcal}. On the other way, a direct estimates show that  $\|\varphi_n(\alpha,\beta)\|^2\geq 1+n|\overline\alpha-\beta|^2$, $\forall n \geq 0$. Hence, uniform boundedness is compatible only with $\overline\alpha=\beta$, and we go back to the trivial situation.
Moreover, since $\|V_\varphi(\alpha,\beta)\varphi_n\|^2=\|\varphi_n(\alpha,\beta)\|^2\geq 1+n|\overline\alpha-\beta|^2$, then $V_\varphi(\alpha,\beta)$ is, in general, unbounded, as already stated. Hence, $\F_\varphi$ and $\F_\Psi$ cannot be Riesz bases since, \cite{you},  two biorthogonal bases can be Riesz bases if and only if they are connected by a bounded operator with bounded inverse.

\subsection{Coherent states}

We now construct the coherent states associated to the model discussed in this section, working first in the coordinate representation. For that, calling $z=z_r+iz_i$, $z_r, z_i\in{\Bbb R}$, and $a=\frac{1}{\sqrt{2}}\left(x+\frac{d}{dx}\right)$,  the normalized solution of the eigenvalue equation $a\eta(x;z)=z\eta(x;z)$,  is, with a certain choice of phase in the normalization,
$\eta(x;z)=\frac{1}{\pi^{1/4}}\exp\left\{-\frac{x^2}{2}+\sqrt{2}\,z\,x-z_r^2\right\}$. Hence, calling $\Phi_\alpha(x;z)$ the eigenstate of $A_\alpha$ with eigenvalue $z$, $A_\alpha\Phi_\alpha(x;z)=z\Phi_\alpha(x;z)$, we get $\Phi_\alpha(x;z)=\eta(x;z+\alpha)$. Analogously, the eigenstate of $B_\beta^\dagger$ with eigenvalue $z$, $B_\beta^\dagger\Psi_\beta(x;z)=z\Psi_\beta(x;z)$, is  $\Psi_\beta(x;z)=\eta(x;z+\overline\beta)$. It is clear that $$\frac{1}{\pi}\int_{\Bbb{C}} \,dz\, |\Phi_\alpha(x;z)><\Phi_\alpha(x;z)|=\frac{1}{\pi}\int_{\Bbb{C}}\,dz\, |\Psi_\beta(x;z)><\Psi_\beta(x;z)|=\1.$$
On the other hand, taken $f,g\in\Hil$, we get
$$
\left<f,\left(\frac{1}{\pi}\int_{\Bbb{C}}\,dz\, |\Phi_\alpha(x;z)><\Psi_\beta(x;z)|\right)g\right>=
e^{-(\alpha_r-\beta_r)^2/2}\,\int_{\Bbb{R}}\,dx\,\overline{f(x)}\,g(x)e^{i\sqrt{2}\,(\alpha_i+\beta_i)},
$$
with obvious notation. Therefore, if $\alpha\neq\overline \beta$, the integral over $\Bbb C$ above does not produce the identity operator! The same conclusion can be recovered working as in Section IV. Following (\ref{43}) we  introduce
\be
\tilde U_{\alpha,\beta}(z)=\exp\left\{z\,B_\beta-\overline z A_\alpha\right\},\qquad \tilde V_{\alpha,\beta}(z)=\exp\left\{z\,A_\alpha^\dagger-\overline z B_\beta^\dagger\right\},
\label{65}\en
and two associated vectors
$$
\tilde\varphi_{\alpha,\beta}(z)=\tilde U_{\alpha,\beta}(z)\varphi_0, \qquad \tilde\Psi_{\alpha,\beta}(z)=\tilde V_{\alpha,\beta}(z)\varphi_0.
$$
They satisfy $A_\alpha\tilde\varphi_{\alpha,\beta}(z)=z\tilde\varphi_{\alpha,\beta}(z)$ and $B_\beta^\dagger\tilde\Psi_{\alpha,\beta}(z)=z\tilde\Psi_{\alpha,\beta}(z)$, as expected. However we find
$$
\frac{1}{\pi}\int_{\Bbb{C}}\,dz\, |\tilde\varphi_{\alpha,\beta}(z)><\tilde\Psi_{\alpha,\beta}(z)|=U(\alpha)\left(
\frac{1}{\pi}\int_{\Bbb{C}}\,dz\, |\varphi_0(z)><\varphi_0(z)|\,e^{z(\overline\alpha-\beta)+\overline z(\overline\beta-\alpha)}\right)
U(\overline\beta)^\dagger
$$
which returns $\1$ if $\alpha=\overline\beta$, but not otherwise. This is a particular case of a general result:

\begin{prop}
If $\F_\varphi^{(\alpha,\beta)}$ and $\F_\Psi^{(\alpha,\beta)}$ are Riesz bases and biorthogonal then, defining $\tilde\varphi_{\alpha,\beta}(z)$ and $\tilde\Psi_{\alpha,\beta}(z)$ as above, they satisfy the resolution of the identity
$\frac{1}{\pi}\int_{\Bbb{C}}\,dz\, |\tilde\varphi_{\alpha,\beta}(z)><\tilde\Psi_{\alpha,\beta}(z)|=\1$.
\end{prop}

\begin{proof}

Since $\F_\varphi^{(\alpha,\beta)}$ and $\F_\Psi^{(\alpha,\beta)}$ are Riesz bases there exists an (unique) orthonormal basis of $\Hil$, $\{\Phi_n\}$, and two bounded operators with  bounded inverses, $X_{\alpha,\beta}$ and $Y_{\alpha,\beta}$, such that $\varphi_n(\alpha,\beta)=X_{\alpha,\beta}\,\Phi_n$ and $\Psi_n(\alpha,\beta)=Y_{\alpha,\beta}\,\Phi_n$, for all $n\geq0$. Due to the biorthogonality of the two sets we must have $Y_{\alpha,\beta}=(X_{\alpha,\beta}^{-1})^\dagger$. Hence our claim follows easily.

\end{proof}

This Proposition is in agreement with our previous conclusions: we have first seen that $\F_\varphi^{(\alpha,\beta)}$ and $\F_\Psi^{(\alpha,\beta)}$ are not Riesz bases. But they are biorthogonal. Hence the resolution of the identity for the associated coherent states cannot be satisfied!

It is not hard to extend this proposition to the general settings of \cite{bag1}. This will be done in a future paper.

\vspace{1cm}

In this paper we have discussed a general strategy, extending ordinary SUSY quantum mechanics, to construct examples of pseudo-bosons. We have seen how these results are related to PHQM and to coherent states. In particular, an interesting output of our procedure is that it produces many different bases of $\Lc^2(\Bbb R)$ and, under  extra conditions, many examples of Riesz bases.

\section*{Acknowledgements}

The author acknowledges financial support by the Murst. The author also thanks the referee for his useful suggestions, which have improved significantly the paper.

\end{document}